\documentclass[11pt]{article}

\usepackage{amsmath}
\usepackage{amsfonts}
\usepackage{amsthm}
\usepackage{amssymb}
\usepackage[utf8]{inputenc}

\usepackage{textcomp}

\usepackage[margin=1in]{geometry}
\usepackage{indentfirst}

\usepackage{authblk}
\usepackage{abstract}

\usepackage{pbox}
\usepackage{comment}

\usepackage[linesnumbered, boxruled]{algorithm2e}
\usepackage{pgfplots}
\pgfplotsset{compat=1.4}
\usetikzlibrary{shapes}
\usetikzlibrary{plotmarks}

\usepackage{framed}
\usepackage{mdframed}
\usepackage{placeins}
\usepackage{subcaption}
\usepackage{afterpage}
\usepackage{hyperref}
\usepackage{courier}
\usepackage{multirow}
\usepackage[noadjust]{cite}

\def\ruzem{Ruzsa-Szemer{\' e}di}
\newcommand{\Szemeredi}{Szemer{\' e}di}
\newcommand{\eps}{\varepsilon}

\newcommand{\bee}{\mathcal{B}}

\newcommand{\rr}{\mathbb{R}}
\DeclareMathOperator{\dist}{dist}
\DeclareMathOperator{\rs}{\mbox{\tt RS}}

\DeclareMathOperator{\diam}{diam}

\newtheorem{theorem}{Theorem}

\newtheorem{lemma}{Lemma}
\newtheorem{definition}{Definition}

\newtheorem{oq}{Open Question}

\newtheorem*{definition*}{Definition}
\newtheorem*{theorem*}{Theorem}

\newtheorem*{lemma*}{Lemma}

\newcommand{\specialcell}[2][c]{%
  \begin{tabular}[#1]{@{}c@{}}#2\end{tabular}}

\title{New Results on Linear Size Distance Preservers\footnote{A preliminary version of this paper appeared in the conference proceedings of SODA 2017, under the title ``Linear Size Distance Preservers."}}
\author{Greg Bodwin\thanks{bodwin@umich.edu} }
\affil{University of Michigan EECS}
\date{}

\begin{document}



\maketitle
\thispagestyle{empty}

\begin{abstract}
Given $p$ node pairs in an $n$-node graph, a \emph{distance preserver} is a sparse subgraph that agrees with the original graph on all of the given pairwise distances.
We prove the following bounds on the number of edges needed for a distance preserver:
\begin{enumerate}
\item Any $p$ node pairs in a directed weighted graph have a distance preserver on $O(n + n^{2/3} p)$ edges.

\item Any $p = \Omega\left(\frac{n^2}{\rs(n)}\right)$ node pairs in an undirected unweighted graph have a distance preserver on $O(p)$ edges, where $\rs(n)$ is the \ruzem\ function from combinatorial graph theory.

\item As a lower bound, there are examples where one needs $\omega(\sigma^2)$ edges to preserve all pairwise distances within a subset of $\sigma = o(n^{2/3})$ nodes in an undirected weighted graph.
If we additionally require that the graph is unweighted, then the range of this lower bound falls slightly to $\sigma \le n^{2/3 - o(1)}$.
\end{enumerate}
\end{abstract}
\pagenumbering{arabic}
\setcounter{page}{0}

\pagebreak

\section{Introduction}

We study \emph{distance preservers}, a fundamental graph theoretic tool that has been applied in research on distance oracles \cite{EP16}, spanners \cite{BV15, BV16, Pettie09, EFN17, BCE05, ABP17, AB16soda, AB17jacm}, hopsets \cite{ABP17, HP18ipl, HP18swat}, shortcutting \cite{HP18swat, ABP17}, and graph algorithms \cite{Alon02, CGMW18}.
Distance preservers have also been a recently popular object of study in their own right \cite{BCE05, CE06, BV16, Bodwin19, GR17, CGMW18}.
Distance preservers were first studied by Coppersmith and Elkin \cite{CE06}.

\begin{definition} [Distance Preservers \cite{CE06}]
For a graph $G = (V, E)$ and a set $P \subseteq V \times V$ called ``demand pairs,'' a subgraph $H$ is a {\em distance preserver} of $G, P$ if $\dist_G(s, t) = \dist_H(s, t)$ for all $(s, t) \in P$.
When $P = S \times S$ for some $S \subseteq V$, we say that $H$ is a \emph{subset distance preserver} of $G, S$.
\end{definition}

Shortest path trees are the most well-known example of distance preservers: for an $n$-node graph, when $P = \{s\} \times V$, there is always a distance preserver that is a tree and thus has $O(n)$ edges.
Thus, in some sense, the goal of research on distance preservers is to prove versions of this fact without the structural requirement $P = \{s\} \times V$.
There are two simple lower bounds for distance preservers to keep in mind: 
\begin{enumerate}
\item When $G$ is a path and we have a single demand pair at either end of the path, $G$ is the only distance preserver of itself, so $n-1$ edges are required.

\item When $G$ is an unweighted clique and we have any $p$ demand pairs, any distance preserver must keep the $p$ edges corresponding to the endpoints of the demand pairs.
\end{enumerate}

In this work we prove new upper and lower bounds on the extremal number of edges needed for a distance preserver of $p$ node pairs in an $n$-node graph.
Our results imply significant progress in understanding the parameter ranges where the above simple lower bounds are asymptotically tight, and \emph{linear size distance preservers} on $O(n+p)$ edges are available in general, so this will be a conceptual theme.
We also give some open problems that fit within this theme.

\subsection{Prior Work}

Prior work on distance preservers for general input graphs can be found in \cite{CE06, BV16}.
These papers proved bounds on the number of edges needed for a distance preserver, which are listed in Table \ref{fig:history}.
In particular, Coppersmith and Elkin \cite{CE06} proved the following facts about distance preservers of linear size:
\begin{itemize}
\item For any $p = O(n^{1/2})$ demand pairs in an $n$-node undirected weighted graph, there is a distance preserver on $O(n)$ edges.
Moreover, the parameter range $p = O(n^{1/2})$ cannot be improved, even if attention is restricted to unweighted graphs \emph{and} we assume the structure $P = S \times S$.

\item In $n$-node undirected weighted graphs, there are examples where $p$ demand pairs require $\omega(p)$ edges for a distance preserver, for any $p = o(n^2)$.
If we additionally require the graph to be unweighted, the range of this lower bound falls slightly to $p = O(n^{2 - o(1)})$.

\item In $n$-node undirected weighted graphs, there are examples where $\sigma$ nodes require $\omega(\sigma^2)$ edges for a subset distance preserver, for any $\sigma = o(n^{3/5})$.
If we additionally require the graph to be unweighted, the range of this lower bound falls to $\sigma = o(n^{9/16})$.
\end{itemize}

\begin{table}[t]
\begin{center}
\begin{tabular}{|c|c|c|c|}
\cline{3-4}
\multicolumn{1}{l}{} && \textbf{Upper Bound} & \textbf{Lower Bound} \\
\hline
\multirow{4}{*}{\textbf{Unweighted}} & \textbf{Undirected} & \specialcell{$O\left(n + n^{1/2}p\right)$\\and $O\left(np^{1/3} + n^{2/3}p^{2/3}\right)$} & \multirow{3}{*}{\specialcell{$\Omega\left(n^{2d / (d^2 + 1)} p^{d(d-1)/(d^2 + 1)}\right)$\\ for any int.\ $d \ge 2$, and\\$\Omega\left(n^{10/11}\sigma^{4/11} + n^{9/11}\sigma^{6/11}\right)$}} \\
\cline{2-3}
& \multirow{2}{*}{\textbf{Directed}}  & \multirow{2}{*}{$O\left(np^{1/2}\right)$} &  \\
& & & \\
\hline

\multirow{2}{*}{\textbf{Weighted}} & \textbf{Undirected} & \specialcell{$O\left(n + n^{1/2}p\right)$\\and $O\left(np^{1/2}\right)$} & \multirow{2}{*}{\specialcell{$\Omega\left(n^{2/3} p^{2/3}\right)$, and \\ $\Omega\left(n^{6/7} \sigma^{4/7}\right)$}} \\
\cline{2-3}
& \textbf{Directed}  & $O\left(np^{1/2}\right)$ & \\
\hline
\end{tabular}
\end{center}
\caption{\label{fig:history} State of the art upper and lower bounds prior to this paper for pairwise distance preservers of $|P| = p$ demand pairs, or for subset distance preservers of $|S| = \sigma$ nodes, in an $n$-node graph.  The second upper bound for undirected unweighted graphs is due to Bodwin and Vassilevska Williams \cite{BV16}, and the remaining bounds are all due to Coppersmith and Elkin \cite{CE06}.}
\end{table}

Variants of distance preservers were studied in \cite{BCE05, CGMW18}, and distance preservers for particular classes of input graphs were studied in \cite{GR17, CGMW18}.
There are also related results for \emph{pairwise spanners}, in which distances are preserved with error \cite{CGK13, KV15, Kavitha17, Woodruff06, Parter14, AB16soda, AB17jacm}, and for \emph{fault-tolerant distance preservers} in which the preserving subgraph must be robust to ``failures'' of a few nodes or edges \cite{PP13, Parter15, BGPV17, GK17}.

\subsection{Our Results}

Our first new result is in the setting of directed weighted graphs:

\begin{theorem} \label{thm:dwintro}
For any $n$-node directed weighted graph and set of $p$ demand pairs, there is a distance preserver on $O(n + n^{2/3} p)$ edges.
\end{theorem}

Theorem \ref{thm:dwintro} is equivalent to its ``linear size'' special case that $p = O(n^{1/3})$ demand pairs have a distance preserver on $O(n)$ edges: given a larger set of demand pairs $P$, one can split $P$ into parts $\{P_i\}$ of size $O(n^{1/3})$ each, building a preserver on $O(n)$ edges for each part $P_i$ separately, and then get the bound in Theorem \ref{thm:dwintro} by unioning these partial preservers together.
We remark that the range where $O(n)$-size preservers exist generally is still open:
\begin{oq}
What is the largest $p = p(n)$ so that any $p$ node pairs in an $n$-node directed graph have a distance preserver on $O(n)$ edges?
By Theorem \ref{thm:dwintro} and prior work of Coppersmith and Elkin \cite{CE06}, the answer is in the range $\Omega(n^{1/3}) = p = O(n^{1/2})$.
The answer to this question might differ between the settings of weighted vs.\ unweighted graphs.
\end{oq}

Our second result is in the setting of undirected unweighted graphs:

\begin{theorem} \label{thm:introuu}
For any $n$-node undirected unweighted graph and set of $p$ demand pairs, there is a distance preserver on $O(p + n^2 / \rs(n))$ edges.
\end{theorem}

An equivalent phrasing of this theorem is that, for any $p = \Omega(n^2 / \rs(n))$, any $p$ demand pairs in an $n$-node undirected unweighted graph have a linear size distance preserver on $O(p)$ edges.
Here $\rs(n)$ is the \ruzem\ function from combinatorial graph theory, defined as follows.

\begin{definition} [Induced Matching]
Given a graph $G = (V, E)$, a set of edges $E' \subseteq E$ is an {\em induced matching} if $E'$ is a matching, and there is $S \subseteq V$ such that $E'$ is exactly the edge set of the subgraph induced on $S$.
\end{definition}

\begin{lemma} [Induced Matching Lemma \cite{RS78}]
If $G = (V, E)$ is an $n$-node graph and $E$ can be partitioned into $n$ induced matchings, then $|E| = o(n^2)$.
\end{lemma}

We then define $\rs(n)$ to be the multiplicative deficit from the Induced Matching Lemma; that is, $\rs(n)$ is the smallest value such that every $n$-node graph on $\le n^2 / \rs(n)$ edges can be partitioned into $n$ induced matchings.\footnote{The name $\rs$ comes from the \emph{Ruzsa-\Szemeredi{} Theorem} \cite{RS78}, which is equivalent to the induced matching lemma but often phrased in different language; see \cite{CF13} (introduction and Section 6.2) for some details.}
It remains a major open question to determine the value of $\rs(n)$; the currently-known bounds are
$$2^{\Omega(\log^* n)} = \rs(n) = 2^{O(\sqrt{\log n})}$$
due to Fox \cite{Fox11} and Behrend \cite{Behrend46} (see also \cite{Elkin10}).\footnote{More specifically: the upper bounds on $\rs(n)$ are implied by graphs derived from dense sets of integers without short arithmetic progressions; these latter constructions are the current best ways to construct progression-free sets.}
Theorem \ref{thm:introuu} is conditionally tight in the following sense.
Coppersmith and Elkin \cite{CE06} give examples where $\omega(p)$ edges are needed for a distance preserver of $p = o\left(n^2 2^{-\Theta(\sqrt{\log n})}\right)$ demand pairs in an undirected unweighted graph.
Under the hypothesis that $\rs(n) = 2^{\Theta(\sqrt{\log n})}$, Theorem \ref{thm:introuu} implies that this range of $p$ is tight, up to the constant hidden in the $\Theta$ exponent.
An interesting question left open is whether a similar result can be proved in the setting of directed unweighted graphs:
\begin{oq}
Is there any $p = o(n^2)$ so that every $n$-node directed unweighted graph and set of $p$ demand pairs has a distance preserver on $O(p)$ edges?
\end{oq}

Our last results are lower bounds for subset distance preservers:

\begin{theorem} \label{thm:iswlb}
For any $\sigma = O(n^{2/3})$, there are examples where $|S| = \sigma$ nodes in an $n$-node undirected weighted graph require $\Omega(\sigma n^{2/3})$ edges for a subset distance preserver.
In particular, $\omega(\sigma^2)$ edges are needed when $\sigma = o(n^{2/3})$.
\end{theorem}

\begin{theorem} \label{thm:isulb}
For any integer $2 \le d \le O(\sqrt{\log n})$ and any $\sigma = O(n^{2/3} 2^{-\Theta(\sqrt{\log n})})$, there are examples of $|S|=\sigma$ nodes in an $n$-node undirected unweighted graph where any subset distance preserver has
$$\Omega\left(n^{\frac{2}{d+1}} \sigma^{\frac{(2d+1)(d-1)}{d(d+1)}} 2^{-\Theta(\sqrt{\log n})}\right) \text{ edges.}$$
In particular, $\omega(\sigma^2)$ edges are needed when $\sigma = o\left(n^{2/3} 2^{-\Theta(\sqrt{\log n})}\right)$.
\end{theorem}

While these results polynomially improve the range in which the lower bound is $\omega(\sigma^2)$, this problem still remains open:
\begin{oq}
Is there a constant $c > 0$ so that every $n$-node graph and set of $\sigma = \Omega(n^{1 - c})$ nodes has a subset distance preserver on $O(\sigma^2)$ edges?
\end{oq}
We can neither prove this statement for the setting of undirected unweighted graphs, nor refute it for the general case of directed weighted graphs.
We consider it to be the main open question in the area of distance preservers.

\section{$O(n)$-Sized Preservers for Directed Weighted Graphs}

Here we prove Theorem \ref{thm:dwintro}.
The argument is directly inspired by the proof of $O(n)$-sized distance preservers for undirected graphs by Coppersmith and Elkin \cite{CE06}, and it can be viewed as an adaptation of their method to the directed setting.

\begin{definition} [Shortest Path Tiebreaking Function \cite{BV16}]
In a graph $G$, a {\em shortest path tiebreaking function} is a map $\pi$ from ordered pairs of nodes $(s, t)$ to a shortest $s \leadsto t$ path.
For a set of demand pairs $P$, we will write $\pi(P)$ as a shorthand for the subgraph with edges $\bigcup \limits_{(s, t) \in P} \pi(s, t)$.
\end{definition}

\begin{definition} [Consistency \cite{BV16}]
A tiebreaking function $\pi$ is {\em consistent} if, for all nodes $w, x, y, z \in V$, if $x, y \in \pi(w, z)$ with $x$ before $y$, then $\pi(x, y)$ is a subpath of $\pi(w, z)$.
\end{definition}

\begin{lemma} [Folklore]
For any graph $G$, there is a consistent tiebreaking function $\pi$.
\end{lemma}
\begin{proof} [Proof]
Modify all edge weights in $G$ by adding an independent uniform random variable drawn from the interval $[0, \eps]$, for some parameter $\eps > 0$.
With probability $1$ there will be no more ties between shortest path lengths; additionally, if we choose $\eps$ small enough then it is not possible for any previously non-shortest path to become a shortest path after this reweighting.
Letting $\pi$ select the unique shortest path in the reweighted graph, whenever $x, y \in \pi(w, z)$ as in the definition of consistency, it must be that $\pi(w, z)$ includes the unique shortest subpath $\pi(x, y)$ between $x$ and $y$.
\end{proof}

Let $\pi$ be any consistent tiebreaking function, and our preserver is the subgraph $H := \pi(P)$.
\begin{definition} [Branching Triple]
A {\em branching triple} is a set of three distinct directed edges $\{e_1 = (u_1, v), e_2 = (u_2, v), e_3 = (u_3, v)\}$ with the same head node.
\end{definition}

\begin{lemma} \label{lem:brtrip} $H$ has at most $\binom{p}{3}$ branching triples.
\end{lemma}
\begin{proof}
Suppose for contradiction that there are more than $\binom{p}{3}$ branching triples in $H$, and so by the pigeonhole principle there exist two branching triples
$$t = \{e_1 = (u_1, v), e_2 = (u_2, v), e_3 = (u_3, v)\} \text{  and  } t' = \{e'_1 = (u'_1, v'), e'_2 = (u'_2, v'), e'_3 = (u'_3, v')\}$$
and three pairs $p_1, p_2, p_3 \in P$ such that $e_i, e'_i \in \pi(p_i)$ for each $i \in \{1, 2, 3\}$.
Assume without loss of generality that the edges $e_1, e_2$ precede $e'_1, e'_2$ in $\pi(p_1), \pi(p_2)$ respectively (otherwise, exchange the roles of the indices $\{1, 2, 3\}$ and perhaps also the roles of $t, t'$ so that this holds).
It follows that $v, v' \in \pi(p_1) \cap \pi(p_2)$, with $v$ preceding $v'$ in both $\pi(p_1), \pi(p_2)$.
Since $\pi$ is consistent, this means that $\pi(v, v')$ is a subpath of both $\pi(p_1)$ and $\pi(p_2)$.
Therefore $\pi(p_1)$ and $\pi(p_2)$ use the same edge entering $v'$, and so $e'_1 = e'_2$.
This contradicts the distinctness of the three edges in $t'$.
\end{proof}

\begin{lemma}
A graph with $O(n)$ branching triples has $O(n)$ edges.
\end{lemma}
\begin{proof}
Consider adding $O(n)$ edges one by one to an initially empty graph.
The first and second edge entering any given node do not create any new branching triples.
Each subsequent edge creates at least one new branching triple.
Therefore, the number of edges in the final graph is at most $2n$ more than the number of branching triples.
\end{proof}
From these two lemmas, $H$ has $O(n)$ edges whenever $p = O(n^{1/3})$.
As discussed in the introduction, this is equivalent to the statement that $p$ demand pairs have a distance preserver on $O(n + n^{2/3}p)$ edges, since one can partition a larger set of demand pairs into parts of size $O(n^{1/3})$ each, build a distance preserver on $O(n)$ edges for each part separately, and then union these together.
So this proves Theorem \ref{thm:dwintro}.

\section{$O(p)$-Sized Preservers for Undirected Unweighted Graphs}

We now prove Theorem \ref{thm:introuu}.
All graphs in this section are undirected and unweighted.
For a tree $T$ rooted at a node $s$, it will be convenient to have the edges oriented away from $s$, so that referring to an edge $(x, y) \in T$ specifically means that $x$ is a parent of $y$, that is, $\dist_T(s, x) < \dist_T(s, y)$.
We may also then refer to $x$ as the \emph{tail} and $y$ as the \emph{head} of the edge $(x, y)$.

\begin{definition} [Branching and Parallel Edges]
In a tree $T$ rooted at a node $s$, a \emph{branching edge} is one whose tail node has out-degree $\ge 2$.
The set of branching edges in $T$ is written $\bee(T)$.
Two edges $(x, y), (x', y') \in T$ are \emph{parallel} if $\dist_T(s, x) = \dist_T(s, x')$, $\dist_T(s, y) = \dist_T(s, y')$, and neither of $(x, y), (x', y')$ is a branching edge.
\end{definition}

\begin{definition} [Lazy Tree]
In a graph $G = (V, E)$, a tree $T$ rooted at a node $s$ is \emph{lazy} if, for all pairs of parallel edges $(x, y), (x',y') \in T$, we have $(x, y'), (x', y) \notin E$.
\end{definition}

Informally, a lazy tree tries to delay the point at which two shortest paths in $T_s$ branch apart for as long as possible.
This intuition is made more precise in the following argument, which builds a lazy tree by iteratively modifying $T$ to delay branchings.

\begin{lemma} \label{lem:lazyexists}
For any graph $G = (V, E), s \in V,$ and set of demand pairs $P_s \subseteq \{s\} \times V$, there is a distance preserver of $P_s$ that is a lazy tree (with respect to root node $s$), whose leaves are exactly the endpoints of pairs in $P_s$.
\end{lemma}
\begin{proof}
Initialize $T_s$ to be any tree formed by overlaying consistent shortest paths for the pairs in $P_s$.
While $T_s$ is not yet lazy, by definition there are edges $(x, y), (x', y') \in T_s \setminus \bee(T_s)$ with $\dist_G(s, x) = \dist_G(s, x')$, $\dist_G(s, y) = \dist_G(s, y')$, and also $(x, y') \in E$.
We then modify $T_s$ in two steps:
\begin{itemize}
\item Add $(x, y')$ and remove $(x', y')$ (note this does not change $\dist_{T_s}(s, y')$), and
\item If one can remove any edges from the new tree $T_s$ without changing the fact that it contains a shortest path for each demand pair in $P_s$, then remove those edges.\footnote{We need to remove extra edges as we go, rather than removing them all in one shot at the end of the construction, to avoid a scenario where a branching edge turns into a non-branching edge because the edge it branches with is removed.  This could potentially break the laziness of the tree.}
\end{itemize}

Let us argue that this iterative modification of $T_s$ must eventually halt, at which point $T_s$ is lazy, and (by the second point) only endpoints of pairs in $P_s$ can be leaves.
If at least one edge is removed from $T_s$ in the second step, then the overall number of edges in $T_s$ decreases in this round.
Otherwise, if no edge is removed from $T_s$ in the second step, then the overall number of edges in $T_s$ remains the same but the number of \emph{non-branching} edges in $T_s \setminus \bee(T_s)$ decreases by $2$, since $(x', y')$ is removed and $(x, y), (x, y')$ are now branching edges.
So the first case happens $O(n)$ times in total, and the second case happens at most $O(n)$ times between instances of the first case.
\end{proof}

Now, given $G=(V, E)$ and demand pairs $P$, we build a distance preserver $H = (V, E_H)$ by partitioning the demand pairs into $n$ subsets $\{P_s\}$, with each $P_s \subseteq \{s\} \times V$ for some $s \in V$.
Using Lemma \ref{lem:lazyexists}, we build a lazy tree $T_s$ to preserve distances for each $P_s$, and finally we union these trees together to get $H$.
We now bound the number of edges in $H$.

\begin{lemma} \label{lem:matchpartition}
We can discard a constant fraction of the edges in $E_H$, and then partition the remaining edges in $E_H \setminus \bigcup \limits_{s \in V} \bee(T_s)$
into $3n$ induced matchings.
\end{lemma}
\begin{proof}
For each node $v \in V$, flip a coin to assign $v$ the label \emph{close} or \emph{far}.
Then, let us say that an edge $(x, y) \in T_s, \dist_{T_s}(s, x) < \dist_{T_s}(s, y)$, is \emph{surviving} in $T_s$ if $x$ was assigned the label \emph{close} and $y$ was assigned the label \emph{far}.
Discard any edges from $E_H$ that are not surviving in any tree $T_s$.
Note that each edge survives with probability at least $1/4$, so in expectation we discard only a constant fraction of the edges in this way.

For each tree $T_s$, partition the surviving edges of $T_s \setminus \bee(T_s)$ into three sets $M_s^{0}, M_s^{1}, M_s^{2}$.
Each edge $(x, y) \in T_s \setminus \bee(T_s)$ is assigned to the set $M_s^i$ where $i := \dist_G(s, x)$ mod $3$.
We now argue that each set $M_s^i$ is an induced matching in $H$.
Let $(x, y), (x', y') \in M_s^{i}$ with $\dist_{T_s}(s, x) < \dist_{T_s}(s, y)$ and $\dist_{T_s}(s, x') < \dist_{T_s}(s, y')$.
There are two cases:
\begin{itemize}
\item Suppose that $\dist_G(s, x) \ne \dist_G(s, x')$.
Since we have $\dist_G(s, x) = \dist_G(s, x')$ mod $3$, without loss of generality we can write $\dist_G(s, x) \ge 3 + \dist_G(s, x')$.
By the triangle inequality, it follows that $(x, x'), (x, y'), (x', y) \notin E$, since any one of these edges would imply an $x' \leadsto x$ path in $G$ of length $1$ or $2$.
We also have $(y, y') \notin E$ by the triangle inequality, since $\dist_G(s, y) \ge 3 + \dist_G(s, y')$.

\item Suppose instead that $\dist_G(s, x) = \dist_G(s, x')$, and thus $(x, y), (x', y')$ are parallel.
It follows from the definition of lazy trees that $(x, y'), (x', y) \notin E$.
Additionally, since $(x, y), (x', y')$ both survived in $T_s$, it must be that $x, x'$ are both labeled \emph{close} and $y, y'$ are both labeled \emph{far}.
Thus $(x, x'), (y, y') \notin E_H$ since the endpoints of these edges have the same labels, so they would be discarded in the first step.
\end{itemize}
Thus each $M_s^i$ is an induced matching, and there are $3n$ such sets in total over all $n$ possible choices of $s$.
\end{proof}

To prove Theorem \ref{thm:introuu}, we now count
$$\left| \bigcup \limits_{s \in V} \bee(T_s) \right| \le \sum \limits_{s \in V} \left| \bee(T_s) \right| = \sum \limits_{s \in V} O(|P_s|) = O(p).$$

Then, by Lemma \ref{lem:matchpartition}, after discarding half the edges in $H$ the remaining edges in $E_H \setminus \bigcup_{s \in V} \bee(T_s)$ can be partitioned into $3n$ induced matchings.
By discarding the $2n$ induced matchings of smallest cardinality, we lose at most an additional $\frac{2}{3}$ fraction of the edges and get a graph whose edges can be partitioned into $n$ induced matchings.
It follows that
$$\left| E_H \setminus \bigcup_{s \in V} \bee(T_s) \right| = O\left(\frac{n^2}{\rs(n)} \right).$$
Putting these together, we have
$$\left|E_H \right| = \left| \bigcup \limits_{s \in V} \bee(T_s) \right| + \left| E_H \setminus \bigcup_{s \in V} \bee(T_s) \right| = O(p) + O\left(\frac{n^2}{\rs(n)} \right),$$
which completes the proof of Theorem \ref{thm:introuu}.

\section{Lower Bounds for Subset Preservers}

We now prove Theorems \ref{thm:iswlb} and \ref{thm:isulb}.
All graphs in this section are undirected (but we will specify weighted/unweighted in context).

\subsection{The Obstacle Product}

The \emph{obstacle product} is a framework for manipulating systems of paths in graphs, which has appeared repeatedly in prior work (\cite{AB17jacm, ABP17, HP18swat}, etc).
Here we will not present the obstacle product in its full generality, rather, we just describe the special case used in this paper.
Recall that a graph $G = (V, E)$ is \emph{$\ell$-layered} if one can partition its nodes $V = V_1 \cup \dots \cup V_{\ell}$ such that all edges go between two adjacent parts $V_i, V_{i+1}$.

\begin{definition} [Perfect Paths]
Let $G = (V_1 \cup \dots \cup V_{\ell}, E, w)$ be a layered graph. A set of paths $\Pi$ is \emph{perfect} if each $\pi \in \Pi$ is the unique shortest path between its endpoints, each $\pi$ starts in $V_1$ and ends in $V_{\ell}$ with exactly one node in each layer, and each $e \in E$ is in exactly one $\pi \in \Pi$.
\end{definition}

The \emph{obstacle product} is the following operation.
Let $G, H$ be layered graphs and let $\Pi_G, \Pi_H$ be perfect sets of paths in $G, H$ respectively.
Suppose $G$ has exactly three layers $V_1, V_2, V_3$ (but $H$ could have more).
Additionally, suppose there is a parameter $q$ such that (1) each node in $V_2$ is in exactly $q$ paths in $\Pi_G$, and (2) $|\Pi_H| = q$.
The \emph{obstacle product} generates a new graph and set of paths as follows.
For each $v \in V_2$ perform the following steps:
\begin{itemize}
\item Delete $v$ from $G$ and add a new disjoint copy $H_v$ of $H$.

\item For each path $\pi \in \Pi_G$ that contained $v$, choose a unique path $\pi_H \in \Pi_H$, and replace the instance of $v \in \pi$ with the entire path $\pi_H$ through $H_v$.
(Notice that the unique correspondence between $\Pi_H$ and paths containing $v$ is possible since we have assumed that $v$ is in $q = |\Pi_H|$ paths in $\Pi_G$.)
To ensure that the new replaced path is still contained in the modified graph, we add two new edges:
\begin{itemize}
\item Letting $u$ be the node preceding $v$ in $\pi$, we add an edge connecting $u$ to the first node in $\pi_H$, whose weight is $w_G(u, v)$.
\item Letting $w$ be the node following $v$ in $\pi$, we add an edge connecting $w$ to the last node in $\pi_H$, whose weight is $w_G(v, w)$.
\end{itemize}
\end{itemize}

\begin{figure}[h]
\begin{center}
\begin{tikzpicture}
\draw [ultra thick] (-1.5, 1) -- (1.5, -1);
\draw [ultra thick] (-1.5, 0) -- (1.5, 0);
\draw [ultra thick] (-1.5, -1) -- (1.5, 1);
\draw [fill=black] (0, 0) circle [radius=0.15];
\node at (0, -0.4) {$v$};
\draw [ultra thick] (0, 0) ellipse (0.5 and 2);
\draw [ultra thick] (-1.5, 0) ellipse (0.5 and 2);
\draw [ultra thick] (1.5, 0) ellipse (0.5 and 2);

\node at (0, -2.5) {$V_2$};
\draw [ultra thick, ->] (3, 0) -- (4.5, 0);
\node at (3.75, -0.5) {o.p.};

\begin{scope}[shift={(8, 0)}]
\draw [ultra thick] (0, 0) circle [radius=1];
\node at (0, -1.5) {$H_v$};
\draw [ultra thick] (0, 0) ellipse (1 and 2);
\draw [ultra thick] (2, 0) ellipse (0.5 and 2);
\draw [ultra thick] (-2, 0) ellipse (0.5 and 2);

\draw [ultra thick] (-2, 1) -- (-0.6, 0.6);
\draw [ultra thick] (-2, -1) -- (-0.6, -0.6);
\draw [ultra thick] (-2, 0) -- (-0.8, 0);
\draw [ultra thick] (2, -1) -- (0.6, -0.6);
\draw [ultra thick] (2, 1) -- (0.6, 0.6);
\draw [ultra thick] (2, 0) -- (0.8, 0);
\draw [ultra thick, dotted] (-0.6, 0.6) -- (0.6, -0.6);
\draw [ultra thick, dotted] (-0.6, -0.6) -- (0.6, 0.6);
\draw [ultra thick, dotted] (-0.8, 0) -- (0.8, 0);
\node [align=center] at (0, -2.75) {$V_2$\\ (expanded to several layers)};

\draw [fill=black] (-0.8, 0) circle [radius=0.1];
\draw [fill=black] (0.8, 0) circle [radius=0.1];
\draw [fill=black] (-0.6, -0.6) circle [radius=0.1];
\draw [fill=black] (0.6, 0.6) circle [radius=0.1];
\draw [fill=black] (-0.6, 0.6) circle [radius=0.1];
\draw [fill=black] (0.6, -0.6) circle [radius=0.1];
\node at (0, -0.7) {$\Pi_H$};

\end{scope}
\end{tikzpicture}
\caption{\label{fig:op} In the obstacle product, we replace each $v \in V_i$ with a copy of $H$ as illustrated here.}
\end{center}
\end{figure}

The point of the obstacle product is that it preserves perfect paths, in the following sense.
We will only sketch the proof here, as it is standard in the area and fairly straightforward.

\begin{lemma} [e.g., \cite{AB17jacm, ABP17, HP18swat}] \label{lem:op}
There exists $\eps > 0$ such that, letting $\eps_H$ be the graph $H$ with edge weights rescaled by $\eps$, the output graph $G'$ and set of paths $\Pi_{G'}$ of the above obstacle product on $G, \eps H$ is such that $\Pi_{G'}$ is perfect in $G'$.

Moreover, if $G, H$ are unweighted, then we can choose $\eps = 1$ (we do not need to rescale $H$).
\end{lemma}
\begin{proof} [Proof Sketch]
It follows directly from the construction that $G'$ is a layered graph, that the paths in $\Pi_{G'}$ go from the first to the last layer with one node in each, and that each edge in $G'$ is in exactly one $\pi' \in \Pi_{G'}$.
The remaining part is to prove that each $\pi' \in \Pi_{G'}$ is the unique shortest path between its endpoints.
In the unweighted setting, we can use uniqueness of the original shortest path $\pi \in \Pi_G$ and uniqueness of the replaced shortest subpath $\pi_H$ to argue that the corresponding path $\pi' \in \Pi_{G'}$ is the only path between its endpoints that uses exactly one node in each layer of $G'$, and hence it is a unique shortest path.
In the weighted settting, we notice that the length of the original shortest path $\pi \in \Pi_G$ increases by at most $\eps \cdot \diam(H)$ over the obstacle product; thus, by choosing $\eps$ small enough, we can guarantee that $\pi$ remains shorter than any second-shortest path between the same endpoints, and so it is still a unique shortest path.
\end{proof}

\subsection{Weighted Lower Bound}

Our lower bound constructions are obtained by simply plugging the right choices of $G, H$ into the obstacle product.
We start with the following useful graphs, which are a light modification of constructions that appear in prior work:

\begin{theorem} [\cite{ST83, CE06}] \label{thm:stwlb}
For any integers $n, \ell \le n$, and $x \le n/\ell$, there is an $\ell$-layered weighted graph $G$ with $n$ nodes in each layer and a set of perfect paths $\Pi$ such that each node is in exactly $x$ paths in $\Pi$.
\end{theorem}
\begin{proof}
Let $[n] := \{0, \dots, n-1\}$.
The nodes of $G$ are $[\ell] \times [n]$.
For all integers $i \in [n], a \in [x]$, we include in $\Pi$ the path
$$\pi = \left( (0, i), (1, i+a), \dots, (\ell-1, i+(\ell-1)a) \right)$$
where addition in the second coordinate is performed mod $n$.
The edges of $G$ are exactly the edges that appear in any such path, and the weight of an edge from a path $\pi$ as above is $\sqrt{1 + a^2}$.

It is immediate that $G$ is $\ell$-layered and that each edge is in exactly one path in $\Pi$, so we will prove here that the paths $\pi \in \Pi$ are unique shortest paths.
Notice that each $\pi$ can be considered as a line segment in $\rr^2$, and due to the setting of edge weights the length of $\pi$ is exactly the Euclidean length of the line segment, up to the detail that the space wraps around in one dimension since we use modular arithmetic.
However, since we add mod $n$ and the total gain in the second coordinate is $(\ell-1) a \le (\ell-1) x < n$, it is impossible to \emph{fully} wrap around, so this detail may be ignored.
Thus, each $\pi \in \Pi$ is the unique shortest path between its endpoints since line segments are unique shortest paths in $\rr^2$.
\end{proof}

\begin{proof} [Proof of Theorem \ref{thm:iswlb}]
Using Theorem \ref{thm:stwlb} twice, we have the following graphs:
\begin{itemize}
\item Let $G = (V_1 \cup V_2 \cup V_3, E, w)$ be an $\ell=3$ layered graph on $\sigma$ nodes per layer, let $\Pi_G$ be a set of $|\Pi_G| = \Theta(\sigma^2)$ perfect paths, and suppose each node is in exactly $x_G = \Theta(\sigma)$ paths in $\Pi_G$.

\item Let $H$ be a graph on $\ell_H = \Theta( n^{2/3} / \sigma )$ layers and $n_H = \Theta(n/(\sigma \ell_H))$ nodes per layer, and let $\Pi_H$ be a set of perfect paths where each node in $H$ is in exactly $x_H$ paths in $\Pi_H$.
Its size is then
$$|\Pi_H| = n_H x_H = \Theta\left( \frac{n_H^2}{\ell_H} \right) = \Theta\left( \frac{n^2}{\sigma^2 \ell_H^3} \right) = \Theta(\sigma);$$
in particular, we may thus ensure that $x_G = |\Pi_H|$.
\end{itemize}

Use the obstacle product to replace each node in the middle layer $V_2$ of $G$ with a copy of $\eps H$ (for some small enough $\eps > 0$, as in Lemma \ref{lem:op}), giving a new graph $G'$ and perfect paths $\Pi_{G'}$.
The number of nodes in $G'$ is $\sigma + \sigma n_H \ell_H + \sigma = \Theta(n)$.
The total number of edges in the paths in $\Pi_{G'}$ is
$|\Pi_G| (\ell_H + 1) = \Theta\left(\sigma n^{2/3} \right)$.
Since $\Pi_{G'}$ is perfect and all its paths go from $V_1$ to $V_3$, any subset distance preserver of $V_1 \cup V_3$ must keep all edges in $\Pi_{G'}$.
By construction there are $\Theta(\sigma)$ nodes in $V_1 \cup V_3$, proving the theorem.
\end{proof}

\subsection{Unweighted Lower Bound}

The proof of Theorem \ref{thm:isulb} is essentially the same as that of Theorem \ref{thm:iswlb}, except that the parameters of the relevant graphs are different.
We again start with a light modification of graphs from prior work:
\begin{theorem} [\cite{Behrend46, Alon02, CE06, AB17jacm, HP18swat}, etc] \label{thm:stulb}
For any integers $n$, $d \ge 2, \ell \le n^{1/d}$, and
$$x = O\left(  n^{\frac{d-1}{d+1}} \ell^{-d \frac{d-1}{d+1}}\right),$$
there is an $\ell$-layered unweighted graph $G$ with $n$ nodes in each layer and a set of perfect paths $\Pi$ such that each node is in exactly $x$ paths in $\Pi$.
\end{theorem}
\begin{proof}
For some parameter $r$, let $V_r$ denote a set of vectors in $\mathbb{N}^d$ such that each $v \in V_r$ has nonnegative integer coordinates and Euclidean length $\le r$, and the vectors $V_r$ form a strictly convex set (i.e., there is no way to write any $v \in V_r$ as a linear combination of the rest with nonnegative coefficients that sum to $\le 1$).
It is proved in \cite{BL98} that one can construct a set with these properties of size up to
$$|V_r| = O\left(r^{d \frac{d-1}{d+1}} \right).$$
We use $V_r$ to build a graph as follows.
Again we will use the notation convention $[n] := \{0, \dots, n-1\}$.
The nodes of $G$ are $[\ell] \times [\ell r]^d$.
For every tuple $t \in [\ell r]^d$ and $v \in V_r$, we include in $\Pi$ the path
$$\pi = \left((0, t), (1, t+v), \dots, (\ell-1, t + (\ell-1)v)\right)$$
where addition in the second coordinate, which holds elements of $[\ell r]^d$, is performed coordinatewise mod $\ell r$.
Setting $r = n^{1/d} / \ell$, we have $n$ nodes per layer and $$|V_r| = O\left(r^{d \frac{d-1}{d+1}}\right) = O\left(n^{\frac{d-1}{d+1}} \ell^{-d \frac{d-1}{d+1}} \right).$$
Due to the modular addition, a symmetry argument over the nodes in each layer implies that each node is contained in exactly $x=|V_r|$ paths in $\Pi$.
Additionally, one can argue that there is no alternate path of length $\ell-1$ between the endpoints $(0, t) \leadsto (\ell - 1, t + (\ell-1)v)$ of $\pi$, as any such path would let us express $v$ as a convex combination of other vectors from $V_r$.
%
\end{proof}

\begin{proof} [Proof of Theorem \ref{thm:isulb}]
For ease of notation, we will write $\Theta^*(\cdot)$ to hide factors of the form $2^{\Theta(\sqrt{\log n})}$.
Using Theorem \ref{thm:stulb} twice, we have the following graphs:
\begin{itemize}
\item Let $G = (V_1 \cup V_2 \cup V_3, E)$ be an $\ell=3$ layered graph on $\sigma$ nodes per layer, let $\Pi_G$ be a set of perfect paths where each node is in exactly
$$x_G = \Theta\left( \sigma^{\frac{d-1}{d+1}} 3^{-d \frac{d-1}{d+1}} \right)$$
paths in $\Pi_G$, and hence
$$|\Pi_G| = |V_1| x_G = \Theta\left( \sigma^{1 + \frac{d-1}{d+1}} 3^{-d \frac{d-1}{d+1}} \right).$$
We may assume $\sigma = \text{poly}(n)$, since the lower bound of Theorem \ref{thm:isulb} is trivial (i.e., it is $O(n)$) in the parameter range $\sigma = O(n^{1/3})$.
We then set the parameter $d := \Theta(\sqrt{\log n})$ for these graphs, giving
\begin{align*}
x_G &= \Theta^*\left( \sigma^{1 - 1/\Theta\left(\sqrt{\log n}\right)} \right)\\
&= \Theta^*\left( \sigma 2^{-\Theta\left(\sqrt{\log \sigma}\right)} \right)\\
&= \Theta^*\left(\sigma \right),
\end{align*}
and so $|\Pi_G| = \Theta^*\left( \sigma^2 \right)$.

\item
Let $H$ be a graph on
$$\ell_H := \Theta^* \left(n^{\frac{2}{d+1}} \sigma^{-\frac{3d+1}{d(d+1)}} \right)$$
layers and $n_H := n / (\ell_H \sigma)$ nodes per layer, and let $\Pi_H$ be a set of perfect paths where each node in $H$ is in exactly $x_H$ paths in $\Pi_H$.
We have:
\begin{align*}
x_H &:= \Theta^*\left(n_H^{\frac{d-1}{d+1}} \ell_H^{-d \frac{d-1}{d+1}} \right)\\
&= \Theta^*\left( \left( \frac{n}{\ell_H \sigma} \right)^{\frac{d-1}{d+1}} \ell_H^{-d \frac{d-1}{d+1}} \right)\\
&= \Theta^*\left( \left(\frac{n}{\sigma}\right)^{\frac{d-1}{d+1}} \ell_H^{-d+1}  \right)\\
&= \Theta^*\left( \left(\frac{n}{\sigma}\right)^{\frac{d-1}{d+1}}  n^{-\frac{2d-2}{d+1}} \sigma^{\frac{(3d+1)(d-1)}{d(d+1)}} \right)\\
&= \Theta^*\left( n^{-\frac{d-1}{d+1}} \sigma^{\frac{2d^2 - d - 1}{d(d+1)}} \right).
\end{align*}
Thus $\Pi_H$ has size
\begin{align*}
|\Pi_H| = n_H x_H &= \Theta^*\left( \left(\frac{n}{\ell_H \sigma}\right) \left(n^{-\frac{d-1}{d+1}} \sigma^{\frac{2d^2 - d - 1}{d(d+1)}}\right) \right)\\
&= \Theta^*\left(\ell_H^{-1} n^{\frac{2}{d+1}} \sigma^{1 - \frac{3d+1}{d(d+1)}} \right)\\
&= \Theta^*\left(\sigma\right).
\end{align*}
\end{itemize}
Thus we may have $\Pi_H = x_G$, so we can use the obstacle product to replace each node in the middle layer $V_2$ of $G$ with a copy of $H$, giving a new graph $G'$ and perfect paths $\Pi_{G'}$.
The number of nodes in $G'$ is $\sigma + \sigma n_H \ell_H + \sigma = \Theta(n)$.
The total number of edges in the paths in $\Pi_{G'}$ is
$$|\Pi_G|(\ell_H + 1) = \Theta^*\left( \sigma^2 \left( n^{\frac{2}{d+1}} \sigma^{-\frac{3d+1}{d(d+1)}}  \right) \right) = \Theta^*\left( n^{\frac{2}{d+1}} \sigma^{\frac{(2d+1)(d-1)}{d(d+1)}}\right).$$
Since $\Pi_{G'}$ is perfect and all its paths go from $V_1$ to $V_3$, any subset distance preserver of $V_1 \cup V_3$ must keep all edges in $\Pi_{G'}$.
By construction there are $\Theta(\sigma)$ nodes in $V_1 \cup V_3$, proving the theorem.
\end{proof}

\section*{Acknowledgements}

I am grateful to Amir Abboud, Atish Agarwala, Noga Alon, Virginia Vassilevska Williams, and several anonymous reviewers for useful technical discussions and helpful comments and corrections that have improved the current presentation of this work.

\bibliography{references}
	\bibliographystyle{plain}

\end{document}